\documentclass[a4paper,10pt]{article}

\usepackage{latexsym}
\usepackage{amsmath}
\usepackage{amsthm}
\usepackage{amsfonts}
\usepackage{csquotes}
\usepackage{cite}
\usepackage{mathtools}
\DeclarePairedDelimiter\floor{\lfloor}{\rfloor}
\usepackage{enumitem}
\usepackage{complexity}
\usepackage{authblk}

\usepackage{tikz}
\usetikzlibrary{positioning}
\usetikzlibrary{calc}
\usetikzlibrary{shapes.geometric,decorations,decorations.pathmorphing}
\usetikzlibrary{intersections}
\usetikzlibrary{matrix}

\usepackage{subfig}
\newsavebox{\measurebox}

\tikzset{me/.style={to path={
\pgfextra{% 
 \pgfmathsetmacro{\startf}{-(#1-1)/2}  
 \pgfmathsetmacro{\endf}{-\startf} 
 \pgfmathsetmacro{\stepf}{\startf+1}}
 \ifnum 1=#1 -- (\tikztotarget)  \else
     let \p{mid}=($(\tikztostart)!0.5!(\tikztotarget)$) 
         in
\foreach \i in {\startf,\stepf,...,\endf}
    {%
     (\tikztostart) .. controls ($ (\p{mid})!\i*6pt!90:(\tikztotarget) $) .. (\tikztotarget)
      }
      \fi   
     \tikztonodes
}}}  

\usepackage[ruled,linesnumbered,noend]{algorithm2e}
\usepackage{cleveref}
\newtheorem{theorem}{Theorem}
\newtheorem{lemma}[theorem]{Lemma}
\newtheorem{corollary}[theorem]{Corollary}
\newtheorem{definition}[theorem]{Definition}
\newtheorem{problem}[theorem]{Problem}
\crefname{lemma}{Lemma}{Lemmas}
\crefname{theorem}{Theorem}{Theorems}
\crefname{corollary}{Corollary}{Corollaries}
\crefname{figure}{Figure}{Figures}
\crefname{algorithm}{Algorithm}{Algorithms}

\title{Constant Amortized Time Enumeration of Eulerian trails}
\author[1]{Kazuhiro Kurita}
\author[2]{Kunihiro Wasa}
\affil[1]{National Institute of Informatics, Tokyo, Japan, \texttt{kurita@nii.ac.jp}}
\affil[2]{Toyohashi University of Technology, Aichi, Japan, \texttt{wasa@cs.tut.ac.jp}}
\date{}

\newcommand{\size}[1]{\lvert#1\rvert}
\newcommand{\set}[1]{\left\{#1\right\}}

\newcommand{\boundary}[2]{\partial_{#2}(#1)}
\newcommand{\remain}[2]{\overline{#1_{#2}}}
\newcommand{\cand}[2]{\Gamma_{#2}(#1)}

\newcommand{\children}[1]{\mathcal{C}(#1)}
\newcommand{\order}[1]{\mathcal{O}(#1)}
\newcommand{\multiplicity}[1]{\mu_{#1}}

\newcommand{\rulename}{Contraction Rule}
\newcommand{\rulePendant}{\rulename{}~1}
\newcommand{\ruleBridge}{\rulename{}~2}
\newcommand{\ruleMod}{\rulename{}~3}
\newcommand{\ruleBigPendant}{\rulename{}~4}
\newcommand{\ruleGenTwoGraphs}{Multibirth Rule}

\newcommand{\etal}{\textit{et al.}}

\newcommand{\solsize}{N}

\providecommand{\keywords}[1]{\textbf{\textit{Keywords---}} #1}

\begin{document}

\maketitle

\begin{abstract}
In this paper, we consider enumeration problems for edge-distinct and vertex-distinct Eulerian trails. Here, two Eulerian trails are \emph{edge-distinct} if the edge sequences are not identical, and they are \emph{vertex-distinct} if the vertex sequences are not identical. As the main result, we propose optimal enumeration algorithms for both problems, that is, these algorithm runs in $\mathcal{O}(N)$ total time, where $N$ is the number of solutions. Our algorithms are based on the reverse search technique introduced by [Avis and Fukuda, DAM 1996], and the push out amortization technique introduced by [Uno, WADS 2015]. 

\mbox{}\\
\noindent
\end{abstract}

\section{Introduction}
An \emph{Eulerian trail} in a graph is a trail that visits all edges exactly once.
Finding a Eulerian trail is a classical problem in graph theory. A famous Seven Bridges of K\"onigsberg problem solved by Leonhard Euler in 1736 is one of a well-known application.
Eulerian trails have many other applications such as CMOS circuit design~\cite{Chen:He:Huang:2002} bioinformatics~\cite{Kingsford:Schatz:Pop:2010,Roy:2007}, and automaton theory~\cite{Kari:2003}.  

Deciding for the existence of a Eulerian trail can be done in time polynomial~\cite{Hierholzer:1873}.
However, the counting problem of \emph{edge-distinct} Eulerian trails is \#\P-complete~\cite{Brightwell:Winkler:2005} for general undirected graphs, 
although for directed graphs, Aardenne and Brujin~\cite{Aardenne:Brujin:1951} proposed the BEST algorithm whose running time is polynomial time. 
Here, two Eulerian trails are \emph{edge-distinct} if the edge sequences of them are different. 
Similarly, two Eulerian trails are \emph{vertex-distinct} if the vertex sequences of them are different. If a graph is simple, the set of edge-distinct Eulerian trails and that of vertex-distinct Eulerian trails are equivalent. 
Thus, counting Eulerian trails is intractable for general cases unless $\P = \#\P$. 
Recently, Conte~\etal~\cite{Conte:Grossi:Loukides:Pisanti:Pissis:Punzi:2020} give an algorithm that answers whether a graph contains at least $z$ Eulerian trails in time polynomial in $m$ and $z$, where $m$ is the number of edges in the graph. 

In contrast to counting problems, \emph{enumeration problems} ask to output all the solutions without duplicates. 
Especially, enumeration problems for subgraphs satisfying some constraints have been widely studied, such as spanning trees~\cite{Shioura:Tamura:Uno:1997}, $st$-paths~\cite{Ferreira:Grossi:Rizzi:Sacomoto:Marie:2014,Birmele:Ferreira:Grossi:Marino:Pisanti:Rizzi:Sacomoto:2012}, cycles~\cite{Birmele:Ferreira:Grossi:Marino:Pisanti:Rizzi:Sacomoto:2012}, maximal cliques~\cite{Tomita:Tanaka:Takahashi:2006,Tsukiyama:Ide:Ariyoshi:Shirakawa:1977}, and many others~\cite{Wasa:2016}. 
In this paper, we focus on enumeration problems for Eulerian trails. 
As mentioned in above, because the counting version is intractable, the enumeration problem is also intractable with respect to the size of inputs. 
Hence, in this paper, we aim to develop efficient enumeration algorithms with respect to the size of \emph{both inputs and outputs}. 
Such algorithms are called \emph{output-sensitive algorithms}. 
In particular, an enumeration algorithm $\mathcal{A}$ runs in \emph{polynomial amortized time} if the total running time of $\mathcal{A}$ is $\order{poly(n)\solsize}$ time, where $n$ is the size of inputs and $\solsize$ is the number of solutions. 
That is, $\mathcal{A}$ outputs solutions in $\order{poly(n)}$ time per each on average. 
In this contexts, the ultimate goal is to develop $\order{\solsize}$ time enumeration algorithm, that is $poly(n) \in \order{1}$, and such optimal algorithms have been obtained~\cite{Uno:2015,Schweikardt:Segoufin:Vigny:2018}. 
Under this evaluation, 
Kikuchi~\cite{Kikuchi:2010} proposed an $\order{m\solsize}$ time algorithm for simple general graphs, where $m$ is the number of edges in an input graph.
However, the existence of a constant amortized time enumeration algorithm for Eulerian trails is open. 

In this paper, we propose optimal enumeration algorithms for 
edge-distinct Eulerian trails and vertex-distinct Eulerian trails based on the reverse search technique~\cite{Avis:Fukuda:1996}. 
Intuitively speaking, an enumeration algorithm based on the reverse search technique enumerates solutions by traversing on a tree-shaped search space, called the \emph{family tree}. 
Each node on the tree corresponds to a prefix of a solution called a \emph{partial solution} and each leaf corresponds to some solution. The edge set of the tree is defined by the \emph{parent-child} relation between nodes. in particular, a partial solution $P$ is the \emph{parent} of a partial solution $P'$ if $P'$ is obtained by adding one edge to $P$. 
Although our algorithm is quite simple, with a sophisticated analysis~\cite{Uno:2015} and contracting operations for graphs, we achieve constant amortized time per solution.

\section{Preliminaries}
An undirected graph $G = (V, E)$ is a pair of a vertex set $V$ and an edge set $E \subseteq V\times V$. Note that $(u, v) \in E$ if and only if $(v, u) \in E$ for each pair of vertices $u, v$. 
Graphs may contain self loops and parallel edges.
For each vertex $v$, the neighborhood $N(v, \multiplicity{v})$ is a multiset of the underlying set $N(v)$,
where $N(v)$ is the set of vertices adjacent to $v$ and $\multiplicity{v}: N(v) \to \mathbb N^+$ is a function from $N(v)$ to a positive integer.
That is, $\multiplicity{v}(u)$ represents the number of distinct edges $(v, u)$ in a graph.
The \emph{degree} $d(v)$ of a vertex $v$ is defined as $d(v) = \size{N(v, \multiplicity{v})} = \sum_{u \in N(v)} \multiplicity{v}(u)$.
A vertex $v$ is \emph{pendant} if $d(v) = 1$.
Let $\boundary{v}{G}$ be the set of edges incident to $v$.

A sequence $\pi = (v_1, e_1, v_2, e_2, \dots, v_\ell)$ is a \emph{trail} if for each $i = 1, \dots, \ell - 1$,
$e_i = (v_i, v_{i+1})$ and edges in $\pi$ are mutually distinct.
Note that some vertex appears more than once in a trail.
In particular, a trail $\pi$ is a \emph{path} if $v_i \neq v_j$ in $\pi$ for $i \neq j$. 
A \emph{circuit} is a trail such that the first vertex and the last vertex are equal. 
A trail $\pi$ is an \emph{Eulerian trail} if $\pi$ contains all the edges in $G$.
$G$ is \emph{Eulerian} if $G$ has at least one Eulerian trail.
It is known that $G$ is Eulerian if and only if either every vertex has even degree or exactly two vertices have odd degree. 
In what follows, we assume that input graphs are Eulerian. 
We define $E(\pi) = (e_1, \dots, e_{\ell-1})$ as a subsequence of $\pi$ containing all the edges in $\pi$, and 
similarly, $V(\pi) = (v_1, \dots, v_\ell)$ as a subsequence of $\pi$ containing all the vertices in $\pi$.
Let $\pi_1$ and $\pi_2$ be two Eulerian trails.
We say $\pi_1$ is \emph{edge-distinct} from $\pi_2$ if $E(\pi_1) \neq E(\pi_2)$ and 
we say $\pi_1$ is \emph{vertex-distinct} from $\pi_2$ if $V(\pi_1) \neq V(\pi_2)$.
A sequence $\pi'$ is a \emph{prefix} of $\pi$ if $\pi' = (v_1, e_1, v_2, e_2, \dots, v_{\ell'})$ for some $1 \le \ell' \le \ell$.

Assume that every Eulerian trail starts from $s$ and ends with $t$.
A trail $P$ is a \emph{partial Eulerian trail} in $G$ if $P$ starts from $s$ and there is an Eulerian trail that contains $P$ as a prefix.
Especially, if $E(G) \setminus P \neq \emptyset$, then $P$ is called a \emph{proper partial Eulerian trail}.
An edge $e$ is said to be \emph{addible} if $P+e$ is also a partial Eulerian trail. 
We denote by $t(P)$ the end vertex of a partial Eulerian trail $P$ such that $t(P) \neq s$.
For each partial Eulerian trail $P$,
we denote by $\remain{P}{G} \coloneqq (V', E(G) \setminus P)$ such that a vertex $v$ is contained in $V'$ if there is an edge in $E(G) \setminus P$ such that $v$ is incident to it.
That is, $\remain{P}{G}$ is the set of edges which we have to add to $P$ for constructing a solution.
Now, we formalize our problems as follows. 

\begin{problem}[Edge-distinct Eulerian trail enumeration]
Given a graph $G$, output all the edge-distinct Eulerian trails in $G$ without duplicate. 
\end{problem}

\begin{problem}[Vertex-distinct Eulerian trail enumeration]
Given a graph $G$, output all the vertex-distinct Eulerian trails in $G$ without duplicate. 
\end{problem}

\section{An algorithm for edge-distinct Eulerian trails}
In this section, we propose our enumeration algorithm for edge-distinct Eulerian trails based on the \emph{reverse search} technique proposed by Avis and Fukuda~\cite{Avis:Fukuda:1996}.
We remark that Eulerian trails in this section are considered as a sequence of edges. 
We enumerate solutions by traversing a directed rooted tree, called a \emph{family tree} $\mathcal{T}$.
The tree consists of nodes $\mathcal{P}$ and edges $\mathcal{E}$.
Each node $X$ in $\mathcal{P}$ is associated with a partial Eulerian trail $P$ and a graph $G$.
Let $P(X)$ be the associated partial Eulerian trail of $X$.
If no confusion arises, we identify a node with its associated a partial Eulerian trail.
The \emph{root} of $\mathcal{T}$ is the trail $R$ consisting of $s$.
For two partial Eulerian trails $P_1$ and $P_2$,
there is a directed edge $(P_1, P_2)$ in $\mathcal{E}$ if
$P_1 \subset P_2$ and $\size{P_2 \setminus P_1} = 1$.
We say that $P_1$ is the \emph{parent} of $P_2$ and $P_2$ is a \emph{child} of $P_1$.
Note that a partial Euler path may have more than one children.
Let $\children{P}$ be the set of child partial Eulerian trails of $P$.
From the definition of the parent-child relation,
for any partial Eulerian trail $P$,
there is a unique path from $P$ to the root on $\mathcal{T}$ obtained by recursively removing an edge $(u, t(P))$ in $P$.
Hence, for any node $P$ in $\mathcal{T}$, there is a path from $R$ to $P$ in $\mathcal{T}$. 

Let $\cand{P}{G} \coloneqq \boundary{t(P)}{G} \cap \remain{P}{G}$ be the candidates of addible edges $e$ to $P$ without violating the parent-child relation.
To traverse all the children of a partial Eulerian trail $P$, generating partial Eulerian trails $P+e$ is enough if $e \in \cand{P}{G}$ satisfies the condition given in the next lemma.

\begin{lemma}
    \label{lem:child:iff}
    Let $P$ be a partial Eulerian trail in $G$ and $e$ be an edge in $\remain{P}{G}$.
    Then, $P + e$ is a child of $P$ if and only if
    $\cand{P}{G} = \set{e}$ or $e \in \cand{P}{G}$ is not a bridge.
\end{lemma}
\begin{proof}
    Suppose that $P+e$ is a child of $P$.
    Assume that $\size{\cand{P}{G}} > 1$ and $e$ is a bridge.
    This implies that $G - (P + e)$ is disconnected, and thus, this contradicts the assumption that $P+e$ is a partial Eulerian trail.

    If $\cand{P}{G}$ contains exactly one edge, then the lemma clearly holds.
    Suppose that $\size{\cand{P}{G}} > 1$.
    From the definition of a partial Eulerian trail, $\cand{P}{G}$ contains at most one bridge.
    Hence, $\cand{P}{G}$ contains a non-bridge edge.
\end{proof}

% From the property of an Eulerian trail, we have an easy observation.
% The following implies that at most one edge in $\cand{P}{G}$ does not yield a child.
% \begin{lemma}
%     \label{lem:uniqueness:bridge}
%     Let $P$ be a partial Eulerian trail in $G$.
%     Then, $\boundary{t(P)}{G}$ contains at most one bridge in $G \setminus P$.
% \end{lemma}

Now, we give our proposed enumeration algorithm as follows. 
The algorithm starts with the root partial Eulerian trail $R$.
Then, for each edge $e = (s, u)$ in $\cand{R}{G}$,
the algorithm generates a child partial Eulerian trail $(e)$ of $R$ if $e$ meets the condition by \cref{lem:child:iff}.
The algorithm recursively generates descendant nodes of $R$ by adding edges in the candidate set by depth-first traversal on $\mathcal{T}$.
If the algorithm finds an Eulerian trail,
then the algorithm outputs it as a solution.
The correctness of the algorithm is clear from the construction.
Hence, we obtain the next theorem.
\begin{theorem}
    \cref{alg:edge:distinct} outputs all the edge-distinct Eulerian trails in a given graph once and only once.
\end{theorem}

\subsection{Amortized analysis: achieving constant amortized time }
\begin{algorithm}[t]
    \SetKwInput{KwData}{Input}
    \SetKwInput{KwResult}{Output}
    \SetKwFunction{Rec}{GenChild}
    \SetKwProg{Fn}{Function}{:}{}
    $R \gets$ the root partial Eulerian trail\; 
    Compute $\remain{R}{G}$ from $G$\; 
    \ForEach{edge $e \in \boundary{s}{G}$}{
    \If{$e$ is not a bridge in $G$}{
    \Rec{$R, e$}\; 
    }
    }
    \SetAlgoLined
    \Fn{\Rec{$P'$, $f$}}{ \label{alg:edge:distinct:rec}
    \KwData{Partial Eulerian trail $P'$, edge $f$}
        $P \gets P' + f$\; \label{alg:edge:distinct:P}
        Compute $\remain{P}{G}$ from $\remain{P'}{G}$\; \label{alg:edge:distinct:remainP}
        \eIf{$\remain{P}{G} = \emptyset$}{
            Output $P$\;
        }{
            Find bridges in $\remain{P}{G}$ by Tarjan's algorithm\; 
            \ForEach{edge $e \in \boundary{t(P)}{G} \cap \remain{P}{G}$}{
                \If{$\cand{P}{G} = \set{e}$ or $e$ is not a bridge}{
                    \Rec{$P$, $e$}\; \label{alg:edge:gen:children}
                }
            }
        }
    }
    \caption{Proposed enumeration algorithm for Eulerian trails}
    \label{alg:edge:distinct}
\end{algorithm}

The goal of this section is to develop an optimal enumeration algorithm constant amortized time per solution based on the algorithm given in the previous section.
We summarize our algorithm in \cref{alg:edge:distinct}. 
We first consider the time complexity for generating all the children of a partial Eulerian trail $P$.
Let $m_P = \size{\remain{P}{G}}$.
Finding all bridges in $\remain{P}{G}$ consumes $\order{m_P}$ time by Tarjan's algorithm~\cite{Tarjan:1974}.
Thus, by using these bridges, we can find edges satisfying the condition in \cref{lem:child:iff} in $\order{m_P}$.
In addition, $\remain{P}{G}$ can be computed in constant time by just removing $f$ from $\remain{P'}{G}$. 
Hence, the following lemma holds.

\begin{lemma}
    \label{lem:time:comp:children}
    Each recursive call at \cref{alg:edge:distinct:rec} of \cref{alg:edge:distinct} makes all the child recursive calls in $\order{m_{P}}$ time.  
\end{lemma}

From \cref{lem:time:comp:children}, we need some trick to achieve constant amortized time enumeration since recursive calls near to the root recursive call needs $\order{m}$ time.
In this paper, we employ the \emph{push out amortization} technique as the trick introduced by Uno~\cite{Uno:2015}.

We first introduce some terminology.
Let $P$ be a partial Eulerian trail.
The phrase \enquote{the algorithm \emph{generates} $P$} indicates that the algorithm computes $P$ and $\remain{P}{G}$ at \cref{alg:edge:distinct:P} and \cref{alg:edge:distinct:remainP} of \cref{alg:edge:distinct}, respectively.
We denote by $T(P)$ of $P$ the exact number of computational steps for generating $P$. That is, $T(P)$ does not include the time complexity of descendants of $P$.
Let $\overline{T}(P)$ be the sum of the computation time of generating children of $P$.

Roughly speaking,
the push out amortization technique gives a sufficient condition for proving that an enumeration algorithm runs in $\order{T^*}$ per solution on average, where $T^*$ is the maximum time complexity among leaves. 
This sufficient condition is called the \emph{push out condition} defined as follows.

\begin{definition}
    Let $P$ be a partial Eulerian trail. We call the following equation the \emph{push out condition} for our problem:
    For some two constants $\alpha > 1$ and $\beta \ge 0$,
    $$\overline{T}(P) \ge \alpha T(P) -\beta (\size{\children{P}} + 1)T^*.$$
\end{definition}

Note that we slightly modify the original condition given in \cite{Uno:2015} to fit our purpose.
The intuition behind this condition is that if the computation time of each node is smaller than the sum of that of descendants, then the whole computation time is dominated by the computation time of descendants, that is, the computation time of leaves.
In our case, $T^* = \order{1}$ because $m_P$ on a leaf is constant. 
If each partial Eulerian trail has at least two children, then the condition clearly holds.
However, if there is a partial Eulerian trail $P$ has at most one child, then the condition may not hold
because for the child $P'$ of $P$,  $\remain{P}{G}$ is larger than $\remain{P'}{G}$ and thus $\overline{T}(P)$ can be smaller than $T(P)$.
Hence,
we modify $\remain{P}{G}$ by contracting vertices and edges
so that each non-leaf partial Eulerian trail $P$ has at least two children.
This \emph{contraction} of a graph is performed by applying the follow rules.
Let $v$ be a vertex in $\remain{P}{G}$.
\begin{description}
    \item[\rulePendant{}]
          If $v$ is a pendant, then remove $v$.
    \item[\ruleBridge{}]
          If $v$ is incident to exactly two distinct edges $(u, v), (v, w)$ and $t(P) \neq v$, then remove $v$ and add edge $(u, w)$ to $\remain{P}{G}$.
\end{description}

% !TeX root = ../main.tex
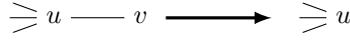
\begin{figure}[t]
    \centering
    \subfloat{
        \begin{tikzpicture}[remember picture]
            \node(EgPendantLeft){
                \begin{tikzpicture}
                    \node(u){$u$};
                    \node[right= 2em of u](v){$v$};
                    \draw (u) edge (v);
                    \node (y1) at ([shift=({160:2em})]u) {};
                    \node (y2) at ([shift=({200:2em})]u) {};
                    \node (y3) at ([shift=({180:2em})]u) {};
                    \draw (u) edge (y1);
                    \draw (u) edge (y2);
                    \draw (u) edge (y3);
                \end{tikzpicture}
            };
        \end{tikzpicture}
    }
    \hspace*{3em}
    \subfloat{
        \begin{tikzpicture}[remember picture]
            \node(EgPendantRight){
                \begin{tikzpicture}
                    \node(u){$u$};
                    \node (y1) at ([shift=({160:2em})]u) {};
                    \node (y2) at ([shift=({200:2em})]u) {};
                    \node (y3) at ([shift=({180:2em})]u) {};
                    \draw (u) edge (y1);
                    \draw (u) edge (y2);
                    \draw (u) edge (y3);
                \end{tikzpicture}
            };
        \end{tikzpicture}
    }
    \tikz[overlay,remember picture]{\draw[-latex,very thick] (EgPendantLeft)-- (EgPendantLeft-|EgPendantRight.west);} 

    \caption{Example of \rulePendant. Vertex $v$ is removed from a graph. }
    \label{fig:eg:rulependant}
\end{figure}
% !TeX root = ../main.tex

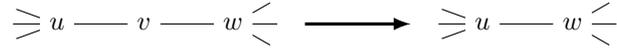
\begin{figure}[t]
    \centering
     \subfloat{
        \begin{tikzpicture}[remember picture]
            \node(EgBridgeLeft){
                \begin{tikzpicture}
                    \node(u){$u$};
                    \node[right= 2em of u](v){$v$};
                    \node[right= 2em of v](w){$w$};
                    \draw (u) edge (v); 
                    \draw (v) edge (w); 
                    \node (x1) at ([shift=({30:2em})]w) {}; 
                    \node (x2) at ([shift=({-30:2em})]w) {}; 
                    \node[label={[name=leftPic]}] (x3) at ([shift=({0:2em})]w) {}; 
                    \draw (w) edge (x1); 
                    \draw (w) edge (x2); 
                    \draw (w) edge (x3); 
                    \node (y1) at ([shift=({160:2em})]u) {}; 
                    \node (y2) at ([shift=({200:2em})]u) {}; 
                    \node (y3) at ([shift=({180:2em})]u) {}; 
                    \draw (u) edge (y1); 
                    \draw (u) edge (y2); 
                    \draw (u) edge (y3); 
                \end{tikzpicture}
                };
        \end{tikzpicture}
        }
        \hspace*{3em}
     \subfloat{
        \begin{tikzpicture}[remember picture]
            \node(EgBridgeRight){
                \begin{tikzpicture}
                    \node(u){$u$};
                    \node[right= 2em of u](w){$w$};
                    \draw (u) edge (w); 
                    \node (x1) at ([shift=({30:2em})]w) {}; 
                    \node (x2) at ([shift=({-30:2em})]w) {}; 
                    \node[label={[name=leftPic]}] (x3) at ([shift=({0:2em})]w) {}; 
                    \draw (w) edge (x1); 
                    \draw (w) edge (x2); 
                    \draw (w) edge (x3); 
                    \node (y1) at ([shift=({160:2em})]u) {}; 
                    \node (y2) at ([shift=({200:2em})]u) {}; 
                    \node (y3) at ([shift=({180:2em})]u) {}; 
                    \draw (u) edge (y1); 
                    \draw (u) edge (y2); 
                    \draw (u) edge (y3); 
                \end{tikzpicture}
                };
            \end{tikzpicture}
        }
    \tikz[overlay,remember picture]{\draw[-latex,very thick] (EgBridgeLeft)-- (EgBridgeLeft-|EgBridgeRight.west);} 
    \caption{Example of \ruleBridge. Vertex $v$ is removed. Note that the degree of other vertices are not changed. } 
    \label{fig:eg:rulebridge}
\end{figure}

See \cref{fig:eg:rulependant,fig:eg:rulebridge}.
Note that when \rulePendant{} is applied to $t$, the last vertex of Eulerian trails becomes the unique neighbor of $t$ in $\remain{P}{G}$. 
We also note that when applying \ruleBridge{}, $u = w$ may hold, that is, a loop on $u$ may be generated.
After recursively applying these operations,
we say $G$ is \emph{contracted} if neither \rulePendant{} nor \ruleBridge{} can not be applicable to $G$.
Clearly, there is a surjection from partial Eulerian trails in $G$ to ones in the contracted graph of $G$.
The next lemma shows the time complexity for obtaining contracted graphs.

\begin{lemma}
    \label{lem:time:get:child}
    Let $G$ be a graph, $P$ be a partial Eulerian trail in $G$ such that there is the parent $P'$ satisfying $P = P' + (u, v)$. 
    Assume that the contracted graph of $\remain{P'}{G}$ is already obtained.
    Then, the contraction of $\remain{P}{G}$ can be done in constant time.
\end{lemma}
\begin{proof}
    The degree of vertices other than $u$ and $v$ is same in $\remain{P'}{G}$ and $\remain{P}{G}$.
    Note that each vertex in a contracted graph has degree at least three.
    Thus, \rulePendant{} is not applicable. 
    Moreover, \ruleBridge{} can be applied to $u$ or $v$ in $\remain{P}{G}$. 
    Hence, by checking the degree of $u$ and $v$, the contraction can be done in constant time.
\end{proof}

% As shown in \cref{lem:time:get:child}, the contraction does not worsen the time complexity for obtaining children of $P$. because the sum $\sum_{(u, v) \in \cand{P}{G}} (d_{\remain{P}{G}}(u) + d_{\remain{P}{G}}(v)) \in \order{m_P}$ holds.
\begin{lemma}
    \label{lem:child:at:least:two}
    Each vertex in a contracted graph is incident to at least two non-bridge edges.
\end{lemma}
\begin{proof}
    Let $G$ be a contracted graph such that $G$ has an Eulerian trail.
    Because of \rulePendant{} and \ruleBridge{}, all vertices are incident to at least three edges.
    Suppose that a vertex $v$ incident to two distinct bridges $b_1$ and $b_2$.
    Otherwise, the lemma holds.
    Let $B_i$ be a connected component of $G - v$ such that $b_i \in B_i$ for $i = 1, 2$.
    Then, there are two classes of Eulerian trails the one contains paths first visit $B_1$ and next visit $B_2$, and the other first visits $B_2$ and next $B_1$.
    Because $G$ is Eulerian, $v$ must be incident to at least two edges for visiting all the edges in $G \setminus (B_1 \cup B_2)$.
    Hence, the statement holds.
\end{proof}

Because detecting bridges can be done in $\order{m_P}$ time~\cite{Tarjan:1974},  the time complexity of each iteration is $\order{m_P}$. 
In addition, the difference between the number of edges in $\remain{P}{G}$ and that in $\remain{Q}{G}$ of a child $Q$ of $P$ is constant, and each proper partial Eulerian trail has at least two children by \cref{lem:child:at:least:two}. 
Hence, the proposed algorithm satisfies the push out condition and runs in constant amortized time per solution \emph{without outputting}.

Next, we consider the output format of our algorithm.
If we simply output all the solutions na\"ively, then total time complexity is $\order{m\solsize}$.
Hence, to achieve constant amortized time enumeration,
we modify our algorithm so that it only outputs the difference between two consecutive solutions.
Outputting only the difference is a well known technique for reducing the total time complexity in enumeration field~\cite{Tomita:Tanaka:Takahashi:2006,Shioura:Tamura:Uno:1997}

During the execution of the algorithm, the algorithm maintains
two working memories $W$ and $\mathcal{I}$. These are implemented by doubly-linked lists.
$W$ contains the current partial Eulerian trail and $\mathcal{I}$ contains the information of contractions.
When the algorithm finds an edge $e$ such that $P+e$ is a child of $P$, then the algorithm appends $e$ to $W$.
In addition, the algorithm also appends an information $\xi$ to $\mathcal{I}$,
where $\xi$ represents how the algorithm contracts a graph.
$\xi$ forms a set of tuples each of which contains what is added and removed from $G$, and which rule is used.
Then, the algorithm prints $+(e, \xi)$.
We can easily see that the length of each output is constant.
When the algorithm backtracks from $P'$, then the algorithm removes $e$ from $W$ and $\xi$ from $\mathcal{I}$, and prints $-(e, \xi)$.
When the algorithm reaches a leaf, the algorithm prints $\#$ to indicate that a solution is stored in $W$. 
By using the information stored in the working memory, we can get the corresponding Eulerian trail if necessary.
Moreover, it is enough to allocate $\order{m}$ space to each working memory.

\cref{lem:child:iff,lem:child:at:least:two} imply that each partial Eulerian trail has either zero or at least two children. This yields the following corollary. 

\begin{corollary}
    \label{cor:num:solutions}
    The number of partial Eulerian trails generated by the algorithm is $\order{\solsize}$.
\end{corollary}

Remind that partial Eulerian trails corresponds to nodes on $\mathcal{T}$. 
Thus, by \cref{cor:num:solutions},
the total length of outputs is $\order{\solsize}$ and the next theorem holds. 

\begin{theorem}
    One can enumerate all the Eulerian trails in a graph in constant amortized time per solution with linear space and preprocess time in the size of the graph.
\end{theorem}

\section{Algorithm for vertex-distinct Eulerian trails}
In this section, we focus on an enumeration of vertex-distinct Eulerian trails.
In a way similar to the previous section,
we develop an enumeration algorithm based on the reverse search technique.
We first some terminologies. 
A pair $(u, v)$ of vertices is \emph{good} if there are two parallel edges $e_1$ and $e_2$ between $u$ and $v$ such that the corresponding trail of $e_1$ in the original graph differs from that of $e_2$. Otherwise, we say $(u, v)$ is \emph{bad}. 
Note that we can check whether $(u, v)$ is good or bad in constant time.

The main difference from the edge-distinct problem is that
adding parallel edges between a bad pair to a current partial Eulerian trail may generate duplicate solutions. 
Note that two parallel edges between a good pair may yield two children if these parallel edges is generated by \ruleBridge{}. 
This is because an added edge by \ruleBridge{} carries two or more edges. 
To ensure that each proper partial Eulerian trail has at least two children, that is, each vertex $v$ has at least two distinct edges,  we slightly modify an input graph of each recursive call as follows.

\begin{description}
    \item[\ruleMod]
          If $v$ has exactly two neighbors $u$ and $w$ such that $\multiplicity{v}(u) = 1$, $(v, w)$ is bad, $u$ and $w$ are distinct from $v$, and $v$ has no loops, then remove $\multiplicity{v}(w)-1$ edges between $v$ and $w$.
          Then, add a vertex $v'$ and $\multiplicity{v}(w) - 1$ edges between $v'$ and $w$.
    \item[\ruleBigPendant]
          Remove incident edges of a vertex $v$
          if $v$ has no loops, $v$ has exactly one neighbor $u$ such that  $(u, v)$ is bad and $\multiplicity{v}(u) \ge 3$. 
          Then, add $\floor{\frac{\multiplicity{v}(u)}{2}}$ loops to $u$. In addition, add an edge $(u, v)$ if $\multiplicity{v}(u)$ is odd.
\end{description}

% !TeX root = ../main.tex

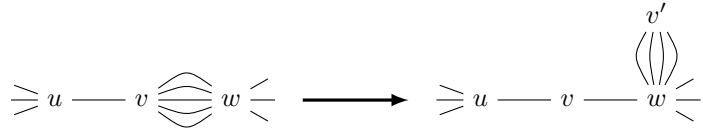
\begin{figure}[t]
    \centering
    \subfloat{
        \begin{tikzpicture}[remember picture]
            \node(EgModLeft){
                \begin{tikzpicture}
                    \node(u){$u$};
                    \node[right= 2em of u](v){$v$};
                    \node[right= 2em of v](w){$w$};
                    \draw (u) edge (v); 
                    \draw (v) edge[me=5] (w); 
                    \node (x1) at ([shift=({30:2em})]w) {}; 
                    \node (x2) at ([shift=({-30:2em})]w) {}; 
                    \node[label={[name=leftPic]}] (x3) at ([shift=({0:2em})]w) {}; 
                    \draw (w) edge (x1); 
                    \draw (w) edge (x2); 
                    \draw (w) edge (x3); 
                    \node (y1) at ([shift=({160:2em})]u) {}; 
                    \node (y2) at ([shift=({200:2em})]u) {}; 
                    \node (y3) at ([shift=({180:2em})]u) {}; 
                    \draw (u) edge (y1); 
                    \draw (u) edge (y2); 
                    \draw (u) edge (y3); 
                \end{tikzpicture}
            };
        \end{tikzpicture}
    }
    \hspace*{3em}
    \subfloat{
        \begin{tikzpicture}[remember picture]
            \node(EgModRight){
                \begin{tikzpicture}
                    \node[label={[name=rightPic]}](u){$u$};
                    \node[right= 2em of u](v){$v$};
                    \node[right= 2em of v](w){$w$};
                    \node[above= 2em of w](vp){$v'$};
                    \draw (u) edge (v); 
                    \draw (w) edge (v); 
                    \draw (vp) edge[me=4] (w); 
                    \node (x1) at ([shift=({30:2em})]w) {}; 
                    \node (x2) at ([shift=({-30:2em})]w) {}; 
                    \node (x3) at ([shift=({0:2em})]w) {}; 
                    \draw (w) edge (x1); 
                    \draw (w) edge (x2); 
                    \draw (w) edge (x3); 
                    \node (y1) at ([shift=({160:2em})]u) {}; 
                    \node (y2) at ([shift=({200:2em})]u) {}; 
                    \node (y3) at ([shift=({180:2em})]u) {}; 
                    \draw (u) edge (y1); 
                    \draw (u) edge (y2); 
                    \draw (u) edge (y3); 
                \end{tikzpicture}
            };
        \end{tikzpicture}
    }
    
    \tikz[overlay,remember picture]{\draw[-latex,very thick] (EgModLeft)-- (EgModLeft-|EgModRight.west);} 
    \caption{Example of \ruleMod. After applying \ruleMod{}, \ruleBridge{} and \ruleBigPendant{} can be applied to the resultant graph. } 
    \label{fig:eg:rulemod}
\end{figure}
% !TeX root = ../main.tex
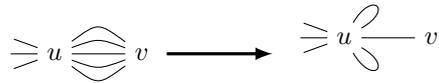
\begin{figure}[t]
    \centering
    \subfloat{
        \begin{tikzpicture}[remember picture]
            \node(EgBigPendantLeft){
                \begin{tikzpicture}
                    \node(u){$u$};
                    \node[right= 2em of u](v){$v$};
                    \draw (u) edge[me=5] (v); 
                    \node (y1) at ([shift=({160:2em})]u) {}; 
                    \node (y2) at ([shift=({200:2em})]u) {}; 
                    \node (y3) at ([shift=({180:2em})]u) {}; 
                    \draw (u) edge (y1); 
                    \draw (u) edge (y2); 
                    \draw (u) edge (y3); 
                \end{tikzpicture}
            };
        \end{tikzpicture}
    }
    \hspace*{3em}
    \subfloat{
        \begin{tikzpicture}[remember picture]
            \node(EgBigPendantRight){
                \begin{tikzpicture}[every loop/.style={}]
                    \node(u){$u$};
                    \node[right= 2em of u](v){$v$};
                    \draw (u) edge (v); 
                    \draw (u) edge[in=60, out=30, loop] (u); 
                    \draw (u) edge[in=-30, out=-60, loop] (u); 
                    \node (y1) at ([shift=({160:2em})]u) {}; 
                    \node (y2) at ([shift=({200:2em})]u) {}; 
                    \node (y3) at ([shift=({180:2em})]u) {}; 
                    \draw (u) edge (y1); 
                    \draw (u) edge (y2); 
                    \draw (u) edge (y3); 
                % \draw[-latex] (leftPic) -> (leftPic-|rightPic.west);
                \end{tikzpicture}
            };
        \end{tikzpicture}
    }
    \tikz[overlay,remember picture]{\draw[-latex,very thick] (EgBigPendantLeft)-- (EgBigPendantLeft-|EgBigPendantRight.west);} 
    \caption{Example of \ruleBigPendant. After applying this rule, \rulePendant{} may be applicable. Note that the degree of $u$ is not changed. }
    \label{fig:eg:rulebigpendant}
\end{figure}

See \cref{fig:eg:rulemod,fig:eg:rulebigpendant}. 
Note that after performing \ruleMod{}, the number of edges in the resultant graph $G'$ and the original graph $G$ are same, and \ruleBridge{} is applicable to $v$ in $G'$.
In addition, after performing \ruleBigPendant{}, \rulePendant{} may be applicable.  
When \ruleMod{} is applicable to $G$,
there are (a) at least one solution containing $(u, v, w)$ or $(w, v, u)$ and (b) at least one solutions containing $(w, v, w)$. 
This implies that if $t(P) = v$ and \ruleMod{} is applicable, then we can generate two partial Eulerian trails which are extended by (a) and (b). 
In addition, we can easily obtain a one-to-one corresponding between solutions in $G$ and $G'$,  which $G'$ is obtained by applying \ruleMod{} or \ruleBigPendant{}.

However, applying \ruleMod{} and \ruleBigPendant{} may make $G$ drastically small.
Let $(v_1, \dots, v_\ell, v_{\ell+1} = v_1)$ be a vertex sequence such that for each $i = 2, \dots, \ell$, $N(v_i) = \set{v_{i-1}, v_{i+1}}$ and $\multiplicity{v_i}(v_{i-1}) = \multiplicity{v_i}(v_{i+1}) = 2$. For each $v_i$, if $t(P) = v_i$, then $P$ has two children. However, after adding an edge incident to $v_i$ to $P$, $v_j$ has at most one child for $j = 2, \dots, \ell-1$. If an input graph forms this sequence, then after contracting by \ruleMod{} and \rulePendant{}, the size of the resultant graph becomes constant.
This shrinking prevents the algorithm from satisfying the push out condition.  
To avoid such a situation, we introduce another rule.

% !TeX root = ../main.tex

\begin{figure}[t]
    \centering
    \sbox{\measurebox}{%
        \begin{minipage}[c][3cm]{.33\textwidth}
            \subfloat[]
            {\label{fig:figA}
                \begin{tikzpicture}[remember picture]
                    \node(EgGenLeft){
                    \begin{tikzpicture}
                        \node(u){$u$};
                        \node[right= 2em of u](v){$v$};
                        \node[right= 2em of v](w){$w$};
                        \draw (u) edge[me=2] (v);
                        \draw (v) edge[me=2] (w);
                        \node (x1) at ([shift=({20:2em})]w) {};
                        \node (x2) at ([shift=({-20:2em})]w) {};
                        \node (x3) at ([shift=({0:2em})]w) {};
                        \draw (w) edge (x1);
                        \draw (w) edge (x2);
                        \draw (w) edge (x3);
                        \node (y1) at ([shift=({160:2em})]u) {};
                        \node (y2) at ([shift=({200:2em})]u) {};
                        \node (y3) at ([shift=({180:2em})]u) {};
                        \draw (u) edge (y1);
                        \draw (u) edge (y2);
                        \draw (u) edge (y3);
                    \end{tikzpicture}
                };
                \end{tikzpicture}
            }
        \end{minipage}}
    \usebox{\measurebox}\quad
    \begin{minipage}[b][\ht\measurebox][s]{.33\textwidth}
        \centering
        \subfloat[]
        {\label{fig:figB}
            \begin{tikzpicture}[remember picture]
                    \node(EgGenRUp){
                    \begin{tikzpicture}
                \node(u){$u$};
                \node[right= 2.8em of u](w){$w$};
                \draw (u) edge[me=2] (w);
                \node (x1) at ([shift=({20:2em})]w) {};
                \node (x2) at ([shift=({-20:2em})]w) {};
                \node (x3) at ([shift=({0:2em})]w) {};
                \draw (w) edge (x1);
                \draw (w) edge (x2);
                \draw (w) edge (x3);
                \node (y1) at ([shift=({160:2em})]u) {};
                \node (y2) at ([shift=({200:2em})]u) {};
                \node (y3) at ([shift=({180:2em})]u) {};
                \draw (u) edge (y1);
                \draw (u) edge (y2);
                \draw (u) edge (y3);
                \end{tikzpicture}
            };
            \end{tikzpicture}
        }

        \vfill

        \subfloat[]
        {\label{fig:figC}
            \begin{tikzpicture}[remember picture]
                    \node(EgGenRDown){
                    \begin{tikzpicture}[every loop/.style={}]
                \node(u){$u$};
                \node[right= 2.8em of u](w){$w$};
                \draw (u) edge [loop right] (u);
                \draw (w) edge [loop left] (w);
                \node (x1) at ([shift=({20:2em})]w) {};
                \node (x2) at ([shift=({-20:2em})]w) {};
                \node (x3) at ([shift=({0:2em})]w) {};
                \draw (w) edge (x1);
                \draw (w) edge (x2);
                \draw (w) edge (x3);
                \node (y1) at ([shift=({160:2em})]u) {};
                \node (y2) at ([shift=({200:2em})]u) {};
                \node (y3) at ([shift=({180:2em})]u) {};
                \draw (u) edge (y1);
                \draw (u) edge (y2);
                \draw (u) edge (y3);
                \end{tikzpicture}
            };
            \end{tikzpicture}
        }
    \end{minipage}
    \tikz[overlay,remember picture]{\draw[-latex,very thick] (EgGenLeft)-- (EgGenLeft-|EgGenRUp.west);} 
    \tikz[overlay,remember picture]{\draw[-latex,very thick] (EgGenLeft)-- (EgGenLeft-|EgGenRDown.west);} 
    \caption{Example of \ruleGenTwoGraphs{}. This modification yields one or two graphs such thse are obtained by replacing (a) with (b), and (a) with (c). }
    \label{fig:eg:rulegentwographs}
\end{figure}
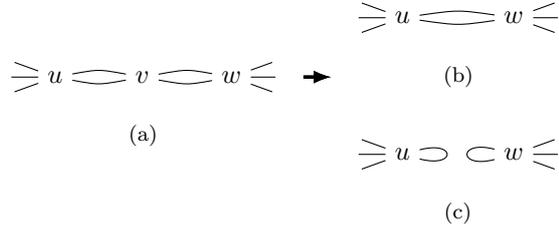
\begin{description}
    \item[\ruleGenTwoGraphs]
          Let $v$ be a vertex adjacent exactly two distinct neighbors $u, w$ such that (1) $t(P) \neq v$, (2) $v$ has no loops, (3) $\multiplicity{v}(u) = \multiplicity{v}(w) = 2$, and (4) both $(u, v)$ and $(v, w)$ are bad.
          Then make a copy $G_1$ of an input graph $G$, and remove $v$ from $G_1$ and add two edges between $u$ and $w$.
          In addition, if $v$ is not a cut in $G$,
          then make an additional copy $G_2$ of an input graph $G$, and
          remove $v$ from $G_2$ and add self loops $(u, u)$ and $(w, w)$.
\end{description}

This rule comes from the following observation: each solution in $G$ contains either (1) two $(u, v, w)$ or (2) $(u, v, u)$ and $(w, v, w)$. Note that if $v$ is a cut, then all the solutions contain the former subtrails. Moreover, the set of solutions in $G_1$ and $G_2$ are disjoint.
Hence, if there is such a vertex $v$, then the algorithm replace a child $Q$ of $P$ with two children $Q_1$ and $Q_2$ whose associated graphs are $G_1$ and $G_2$, respectively.

In this section,
we say that a graph  $G$ is \emph{contracted} if \rulePendant{} to \ruleBigPendant{} can not be applicable to $G \setminus P$.
If a graph is contracted, then we show the key lemmas in this section.

\begin{lemma}
    \label{lem:v:time:contraction}
    Let $G$ be a contracted graph and $P$ be a partial Eulerian trail in $G$.
    Suppose that $\remain{P'}{G}$ of the parent $P'$ of $P$ is already obtained. 
    Then, contracting $\remain{P}{G}$ can be done in constant time. 
\end{lemma}
\begin{proof}
Let $e = (u, v)$ be an edge such that $P = P' + e$.
Firstly, removing $e$ from $\remain{P'}{G}$ can be done in constant time. 
We consider that which rule can be applied to $\remain{P}{G}$.  
Because $\remain{P'}{G}$ is contracted, we have a case analysis with respect to the multiplicity $\multiplicity{v}$ of edges on $\remain{P'}{G}$ and the number of neighbors of $v$ as follows; 

\begin{description}
    \item[Case A: $v$ has at least three distinct neighbors $u, w_1, w_2$] \mbox{}\\
          If $\multiplicity{v}(u) > 1$ or more than three distinct neighbors, then no rules can be applicable to $v$.
          We assume that $v$ has exactly three distinct neighbors $u, w_1, w_2$ such that $\multiplicity{v}(u) = 1$. 
          \begin{description}
              \item[A.1: $\multiplicity{v}(w_1) \ge 2$ or $\multiplicity{v}(w_2) \ge 2$ ] \mbox{}\\
                    We can not apply any rules to $w_1$ and $w_2$ after adding $e$ to $P'$.
            %   \item[A.2: $\multiplicity{v}(w_1) = \multiplicity{v}(w_2) = 2$ ] \mbox{}\\
            %         We can apply \ruleGenTwoGraphs{} to $v$. After applying it, we can not apply any rules to neither $w_1$ nor $w_2$.
              \item[A.2: $\multiplicity{v}(w_1) = 1$ and $\multiplicity{v}(w_2) \ge 1$]  \mbox{} \\
                    We first apply \ruleMod{} and \ruleBridge{} to $v$. Then, possibly, we also apply \rulePendant{}, \ruleBridge, and \ruleBigPendant{} to the copy of $v$.
                    If $w_2$ is incident to $w_3$ such that $\multiplicity{w_2}(w_3) = 1$, then \ruleBridge{} can be applied to $w_2$. However, no more rule can be applied to neither $w_1$ nor $x$.
          \end{description}
    \item[Case B: $v$ has exactly two distinct neighbors $u, w$]\mbox{}\\
          \begin{description}
              \item[B.1: $\multiplicity{v}(u) \ge 3$] \mbox{}\\
                    If $\multiplicity{v}(u) = 3$ and $\multiplicity{v}(w) = 2$, then this case is the same as case A.1. 
                    Otherwise, we can not apply any rules to $v$ and $w$ after adding $e$ to $P'$.
              \item[B.2: $\multiplicity{v}(u) = 2$ ] \mbox{}\\
                    This case is the same as case A.2 after adding $e$ to $P'$.
              \item[B.3: $\multiplicity{v}(u) = 1$ ] \mbox{}\\
                    In this case, \ruleMod{} can be applied to $v$ on $\remain{P'}{G}$. 
                    Because $\remain{P'}{G}$ is contracted, this derives a contradiction. 
                    Thus, this case does not happen. 
        %   Because $\remain{P'}{G}$ is contracted, $v$ has some loops.
        %   Thus, if $\multiplicity{v}(w) \ge 3$, then no rule can be applicable. 
        %   Moreover, if $\multiplicity{v}(w) = 1$, then this contradicts to $G$ is contracted. 
        %   We assume that $\multiplicity{v}(w) = 2$. 
        %   If $w$ has exactly two distinct neighbors $v$ and $x$, then $\multiplicity{w}(x) \ge 3$. Hence, after adding $e$ to $P$ and applying \ruleBigPendant{} to $v$, no rule can be applied to $w$. 
        %   Suppose that $w$ has exactly three distinct neighbors $v$, $x_1$, and $x_2$. 
        %   After applying \ruleBigPendant{}, $w$ has a loop. Hence, if $\multiplicity{w}(x_1) \ge 2$ or $\multiplicity{w}(x_2) \ge 2$, then no rule can be applied to $w$. On the other hand, if $\multiplicity{w}(x_1) = \multiplicity{w}(x_2) = 1$, then after applying \ruleBigPendant{}, we can also apply \ruleBridge{}. However, no more rule can be applicable to neither $x_1$ nor $x_2$. See \cref{fig:eg:lem}. 
        %   \input{fig/example.lem.tex} 
          \end{description}
\end{description}
    We also have the same cases for $u$. Thus, from the above observation, the number of applying contracting rules is constant.
    Moreover, each rule can be proceeded in constant time.
    Therefore, the lemma holds.
\end{proof}

Let $P' = P + e$ be the parent of $P$.  Assume that $\remain{P}{G}$ is obtained by removing $e$ and repeatedly applying \rulePendant{} to \ruleBigPendant{} to $\remain{P'}{G}$ .
Assume that $\remain{P'}{G}$ has no vertex to which \ruleGenTwoGraphs{} can be applied. 
This yields the number of vertices to which \ruleGenTwoGraphs{} is applicable in $\remain{P}{G}$ is constant. 
Moreover, when making children of $P$, we avoid to generate all the whole copies of $\remain{P}{G}$ by applying \ruleGenTwoGraphs{} to reduce the amount of space usage. In stead of whole copying, we locally modify $v$ that is a target of \ruleGenTwoGraphs{}. 
Hence, the following clearly holds. 

\begin{lemma}
\label{lem:time:gen:two:inputs}
    Let $G$ be a graph and $P$ be a partial Eulerian trail in $G$.
    Suppose that $\remain{P}{G}$ is contracted. 
    Then, children of $P$ that are generated by applying \ruleGenTwoGraphs{} to $\remain{P}{G}$ can be obtained in constant time. %, where $m_P$ is the number of edges in $\remain{P}{G}$. 
\end{lemma}

From the above discussion, we can obtain the following key lemma. 

\begin{lemma}
    A proper partial Eulerian trail $P$ of a contracted graph $G$ has at least two children.
    \label{lem:v:pet:has:two:children}
\end{lemma}
\begin{proof}
    Because $G$ is contracted, each vertex has at least two distinct neighbors.
    If $t(P)$ is incident to two bridges, then this contradicts with the definition of a partial Eulerian trail. Thus, $t(P)$ is incident to at most one bridge.
    In addition, if $t(P)$ is not incident to a bridge, then the lemma clearly holds.

    Suppose that $t(P)$ is incident to a bridge $e = (t(P), u)$.
    If $t(P)$ is adjacent to exactly two vertices, then this implies that we can apply \ruleMod{}. Hence, $t(P)$ is adjacent to more than two distinct vertices and the statement holds.
\end{proof}

By \cref{lem:v:time:contraction,lem:v:pet:has:two:children,lem:time:gen:two:inputs}, each non-leaf node $P$ on the family tree made by the algorithm has at least two children and can be done in $\order{m_P}$ time, where $m_P$ is the number of edges in the contracted graph $\remain{P}{G}$. Remind that the computation time is dominated by detecting bridges on $\remain{P}{G}$. 
Thus, by the same discussion in the previous section, the main theorem in this section is established. 

\begin{theorem}
    All vertex-distinct Eulerian trails can be enumerated in constant amortized time per solution with linear time preprocessing and linear space. 
\end{theorem}

\section*{Acknowledgement}
This work was partially supported by JST CREST Grant Number JPMJCR18K3 and JSPS KAKENHI Grant Numbers 19K20350, JP19J10761, and JP20H05793, Japan.

\bibliography{euler}
\bibliographystyle{plain}

\end{document}